\documentclass[11pt,reqno]{amsart}

\usepackage{amscd,amssymb,amsmath,amsthm}
\usepackage[arrow,matrix]{xy}
\usepackage{graphicx}
\usepackage{color}
\usepackage{cite}
\newtheorem{defn}{Definition}
\newtheorem{lemma}{Lemma}
\newtheorem{pro}{Proposition}
\newtheorem{rk}{Remark}

\numberwithin{equation}{section} 

\newcommand{\bea}{\begin{eqnarray}}
\newcommand{\eea}{\end{eqnarray}}

\def\d{\delta}
\def\a{\alpha}

\def\b{\beta}

\def \d {\delta}

\def\c{\gamma}

\begin{document}
\title[On free energies  of Ising model]{On free energies  of the Ising  model on the Cayley tree}

\author{D. Gandolfo, M.M. Rakhmatullaev,  U. A. Rozikov, J. Ruiz}

 \address{D.\ Gandolfo and J.Ruiz\\Centre de Physique Th\'eorique, UMR 6207,Universit\'es Aix-Marseille
 et Sud Toulon-Var, Luminy Case 907, 13288 Marseille, France.}
\email {gandolfo@cpt.univ-mrs.fr\ \ ruiz@cpt.univ-mrs.fr}

\address{M. \ M. \ Rakhmatullaev\\ Namangan State University,
 Namangan, Uzbekistan.}
 \email{mrahmatullaev@rambler.ru}

 \address{U.\ A.\ Rozikov\\ Institute of mathematics,
29, Do'rmon Yo'li str., 100125, Tashkent, Uzbekistan.}
\email {rozikovu@yandex.ru}

\begin{abstract}
We present, for the Ising model on the Cayley tree, some explicit formulae of the free energies (and entropies) according to boundary conditions (b.c.). They include translation-invariant, periodic, Dobrushin-like b.c., as well as those corresponding to (recently discovered) weakly periodic Gibbs states.  The later are defined through a partition of the tree that induces a 4-edge-coloring. We compute the density of each color.
\end{abstract}
\maketitle



\section{Introduction and definitions}

On non-amenable graphs, not only Gibbs measures but also the free energy (and the entropy) depend on the boundary conditions.

The purpose of this paper is to study  this dependence
for  one of the simplest such graph,  the Cayley tree
(Bethe lattice). Our analysis is restricted to the Ising model.

Let $\Gamma^k= (V , L)$ be the uniform Cayley tree, where each vertex has $k + 1$ neighbors with $V$ being the set of vertices and $L$ the set of edges.

On this tree, there is a natural distance to be denoted $d(x,y)$,
 being
 the number of nearest neighbor pairs  of the minimal path between  the vertices $x$ and $y$
 (by  path one means   a collection of  nearest neighbor pairs, two consecutive pairs
 sharing at least a given vertex).

 The Ising model is  defined by the
 formal Hamiltonian
\begin{equation}
\label{h}
H(\sigma)=-J\sum_{\langle x,y\rangle\subset V}
\sigma(x)\sigma(y),
\end{equation}
where
the sum runs over   nearest neighbor vertices
$\langle x,y\rangle$ and the spins $\sigma(x)$ take values in the set  $\Phi=\{ +1,- 1\}$.

For a fixed $x^0\in V$, the root,
let
\begin{equation*}
W_n=\{x\in V :\, d(x,x^0)=n\}, \ \
V_n=\{x\in V :\, d(x,x^0)\leq n\}
\end{equation*}
be respectively the sphere and the ball
of radius $n$ with center at $x^0$,  and for $ x\in W_n$ let
 $$S(x)=\{y\in W_{n+1} :  d(y,x)=1\},$$ be
 the set of direct successors of $x$.

 The (finite-dimensional) Gibbs distributions at
 inverse temperature   $\beta=1/T$
 are defined by
 \begin{equation}\label{*}
\mu_n(\sigma_n)=Z^{-1}_n
\exp\Big\{\beta J
\sum_{\langle x,y\rangle\subset V_n} \sigma(x)\sigma(y)
+\sum_{x\in W_n}h_x\sigma(x)\Big\},
\end{equation}
with  partition functions given by
 \begin{equation}\label{pf}
Z_n
\equiv
Z_n(\b, h)
=\sum_{\sigma_n \in \Phi^{V_n}}
\exp\Big\{\beta J
\sum_{\langle x,y\rangle\subset V_n} \sigma(x)\sigma(y)
+\sum_{x\in W_n}h_x\sigma(x)\Big\}.
\end{equation}
Here
$$h=\{h_x\in R, x\in V\}$$
is a collection of real numbers that stands for (generalized) boundary condition.

The probability distributions (\ref{*}) are said compatible if for all
$\sigma_{n-1}$
\begin{equation}\label{**}
\sum_{\omega_n\in \Phi^{W_n}}\mu_n(\sigma_{n-1}, \omega_n)=\mu_{n-1}(\sigma_{n-1}).
\end{equation}
It is well known that this compatibility condition is satisfied if and only if for any $x\in V$
the following equation holds
\begin{equation}\label{***}
 h_x=\sum_{y\in S(x)}f(h_y,\theta),
\end{equation}
where
\begin{equation}
\label{****}
\theta=\tanh(\b J), \quad f(h,\theta)={\rm arctanh}(\theta\tanh h).
\end{equation}

Namely,   for any boundary condition
satisfying the functional equation (\ref{***}) there exists a unique Gibbs measure, the correspondence being one-to-one.

We will be interested in the dependence with respect to boundary conditions of the free energy defined as the limit
\begin{equation}\label{fe1}
F(\b, h)=-\lim_{n\to \infty}{1\over \b |V_n|}
\ln Z_n(\b, h),
\end{equation}
where $|\cdot|$ denotes hereafter the cardinality of a set.

A boundary condition satisfying (\ref{***}) will be in the sequel called \emph{compatible}.

The paper is organized as follows.  Section 2 provides the first part of results:  a general formula applied then to various known  boundary conditions
(translation-invariant, Bleher-Ganikhodjaev, Zachary, ART),
and those about entropy.
 Periodic and weakly periodic  cases are the subject of
 Section 3.
A first appendix concerns the density of colors mentioned in the abstract.
A second appendix
provides a sufficient condition for the existence of the free energy  in case of compatible boundary conditions.

\section{General formula  and first results}

For compatible  boundary conditions,
the free energy is given by the
formula
\begin{equation}\label{fe2}
F(\b, h)=-\lim_{n\to\infty}{1\over |V_n|}
\sum_{x\in V_{n}}a(x),
\end{equation}
where
$$
a(x)
={1\over 2\beta}\ln[4\cosh(h_x-\b J)\cosh(h_x+\b J)].
$$

To see it,  first notice  that
\begin{equation}\label{re}
Z_n(\beta, h)=A_{n-1}Z_{n-1}(\beta, h),
\end{equation}
where $A_n=\prod\limits_{x\in W_n}b(x)$ with  $b(x)$ satisfying
\begin{equation}\label{b}
\prod_{y\in S(x)}\sum_{u=\pm 1}\exp(\b J\varepsilon u+uh_y)=b(x)\exp(\varepsilon h_x), \ \ \varepsilon=\pm 1.
\end{equation}
This formula used with both values of $\varepsilon$ implies
$$b(x)
=
\prod_{y\in S(x)}\left(4\cosh(h_y-\b J)\cosh(h_y+\b J)\right)^{1/2}
=\exp\Big(\beta\sum_{y\in S(x)}a(y)\Big).$$
It is then enough to insert this formula into the recursive equation (\ref{re}) to  get by iteration
$$Z_n(\b,h)=\prod_{x\in V_{n-1}}b(x) $$
which gives (\ref{fe2}).

Notice  that
\begin{equation}\label{fef}
F(\b, h)=F(\b, -h),
\end{equation}
 where $-h=\{-h_x, x\in V\}$.

\subsection{Translation-invariant boundary conditions}
They correspond to constant functions, $h_x=h$, in which case
 the condition (\ref{***}) reads
\begin{equation}\label{f}
h=kf(h,\theta).
\end{equation}
 The equation (\ref{f}) has a unique solution $h=0$, if $\theta\leq \theta_{\rm c}={1\over k}$ and three distinct solutions $h=0,\pm h_*$ ($h_*>0$), when $\theta >\theta_{\rm c}$.

Let us denote by $\mu_0$, $\mu_{\pm}$ the corresponding Gibbs measures and recall the following known results for the ferromagnetic Ising model ($\theta \geq 0$):

 \begin{itemize}
    \item[(1)] If $ \theta\leq \theta_{c}$,  $\mu_0$ is unique and extreme.

    \item[(2)] If $\theta > \theta_{c}$,   $\mu_-$ and $  \mu_+$, are extreme.

    \item[(3)]  $\mu_0$ is
extreme if and only if $ \theta < {1\over\sqrt{k}}$.
    \end{itemize}

 (see e.g. \cite{Pr}, \cite{Ge}, \cite{BRZ})

According to formula (\ref{fe2}), the free energies of translation-invariant (TI)  b.c. are given by:
\begin{eqnarray}
F_{\rm TI}(\b, 0)
&=&
-{1\over\beta}\ln(2\cosh(\b J)).
\label{zero}
\\
F_{\rm TI}(\b, h_*)
&=&
F_{\rm TI}(\b,-h_*)
=
-{1\over 2\beta}\ln[4\cosh(\b J-h_*)\cosh(\b J+h_*)].
\label{tife}
\end{eqnarray}
Some particular plots are shown in Fig. \ref{fig1}.

\begin{figure}
  \includegraphics[width=11cm]{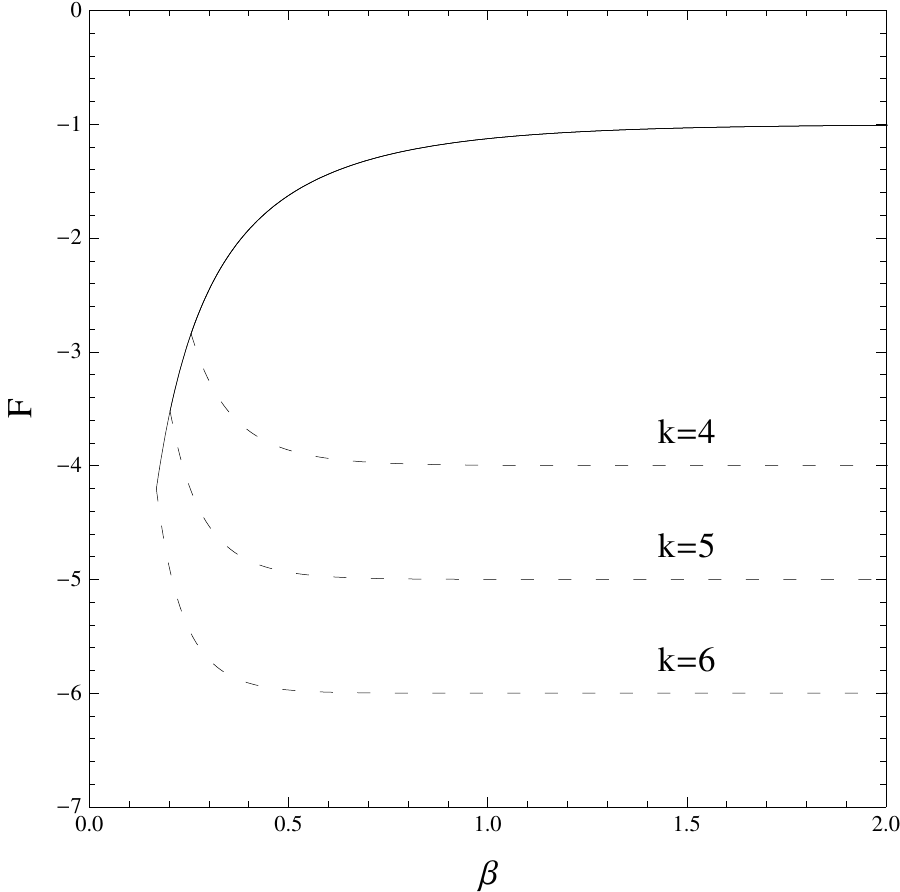}\\
  \caption{
The free energies $F_{\rm TI}(\b, 0)$ (solid line) and $F_{\rm TI}(\b, h_*)$ (dotted lines)
for  $\theta \geq \theta_{c}$  with  $J=1$ and $k=4, 5, 6$.
 }
\label{fig1}
\end{figure}

In order to draw the  free energy
$F_{TI} (\beta, h)$ as a
function of $\beta$, we  notice that
the equation (\ref{f}) gives
\begin{equation}
\label{betah}
\beta(h) = \frac{1}{2 J} \ln \frac{e^{(1+\frac{1}{k})2 h}-1}{e^{2 h}-e^{\frac{2 h}{k}}}
\end{equation}
and   use the  parametric representation defined by the following mapping:
\begin{equation*}
  h \to (\beta (h),\, F(h)), \hspace{1em} \mbox{for} \ \ h \geq 0.
\end{equation*}
The function  $F(h)$ is defined by inserting (\ref{betah}) into (\ref{zero})  and (\ref{tife}).

\subsection{Bleher-Ganikhodjaev construction}
Here one consider the half tree. Namely the root $x^0$
has $k$ nearest neighbours.
Consider an infinite path $\pi=\{x^0=x_0<x_1<\dots\}$  (the notation $x<y$ meaning that  pathes from the root to $y$ go through $x$).
Associate  to this  path
a collection $h^\pi$ of numbers
 given  by the condition
\begin{equation}\label{b3.1}
h_x^\pi=\left\{\begin{array}{ll}
-h_*, \ \ \mbox{if} \ \ x\prec x_n, \, x\in W_n,\\[2mm]
h_*, \ \ \mbox{if} \ \ x_n\prec x, \, x\in W_n,
\end{array}\right.
\end{equation}
$n=1,2,\dots$ where $x\prec x_n$ (resp. $x_n\prec x$) means that $x$ is on the left (resp. right) from the path $\pi$.

For any infinite path $\pi$, the collection of numbers $h^\pi$ satisfying relations (\ref{***}) exists and is unique (see \cite{BG}).

A real number
$t=t(\pi)$, $0\leq t\leq 1$ can be assigned to the infinite path and the set $h^{\pi(t)}$
is uniquely defined.
The  set of numbers $h^{\pi(t)}$ being distinct for different $t\in [0,1]$, it is also the case for the
 corresponding Gibbs  measures. One thus obtains uncountable many  Gibbs measures and they are extreme.

Since   the solution $h^{\pi}$
can differ from  $h_*$ or $-h_*$ only on the path $\pi$,
it is thus obvious that the values $h_x^{\pi}$, $x\in \pi$ do not contribute to the free energy.
As a consequence, we   get from the evenness
(\ref{tife}), that
the  free energies of  Bleher-Ganikhodjaev and  translation-invariant boundary conditions coincide:
\begin{equation}
F_{\rm BG}(\b, h^{\pi})=F_{\rm TI}(\b, h_*).
\end{equation}

\subsection{Zachary construction.}
This construction provides an
(uncountable) set of distinct functions $h^{(t)}$  satisfying
 condition (\ref{***})
and parameterized by  $t\in (-h_*,h_*)$.
It is assumed here that  $\theta>\theta_{\rm c}$.

Take $t\ne 0$ and  define the sequence
$(t_n)_{n\geq 0}$ recursively by $t_0=t$,
\begin{equation}\label{tn}
t_n=kf(t_{n+1}, \theta), \ \ n\geq 0.
\end{equation}
Since the function $f$ is increasing and maps the interval $(-h_*,h_*)$  into itself, the definition of $t_n$  make sense.
Moreover one can see that $\lim_{n\to\infty}t_n=0$ for each $t_0=t$.

Consider the function $h^{(t)}_x=t_n$ for all $x\in W_n$. This function  satisfies condition (\ref{***}) for any $t$ and
by  construction, distinct $t$  assign distinct  $h^{(t)}$.
The  associated  Gibbs measures  are known to be extreme \cite{Z}.

The corresponding free energies can be written as
$$F_{\rm Zach}(\beta, h^{(t)})=-{1\over 2\beta}\lim_{n\to \infty}{1\over |V_n| }
\sum_{m=0}^{n}|W_m| \
 \tilde{a }(t_m),$$
where
$\tilde{a }(t) =\ln[4\cosh(\b J-t)\cosh(\b J+t)].$

By Stolz-Ces\'aro theorem (see e.g. \cite{K}) applied to the sequences
\begin{equation}
 a_n=  \sum_{m=0}^n|W_m| \tilde{a}(t_m), \quad    b_n=|V_n|
\end{equation}
one has
\begin{multline}
\lim_{n\to \infty}{1\over |V_n|}
\sum_{m=0}^{n}|W_m| \   \tilde{a }(t_m)
 =\lim_{n\to \infty}
\frac{\displaystyle \sum_{m=0}^{n+1} |W_m| \   \tilde{a }(t_m)
- \sum_{m=0}^{n}   |W_m| \   \tilde{a}(t_m)}{|V_{n+1}| -|V_n| }
\\
=\lim_{n\to \infty}
\frac{  |W_{n+1}| \   \tilde{a }(t_{n+1})
}{|W_{n+1}|
}
=\lim_{n\to \infty} \tilde{a }(t_{n+1})
=\ \tilde{a}(0).
\end{multline}
As a consequence
\begin{equation}
 F_{\rm Zach}(\beta, h^{(t)})
= F_{\rm TI}(\beta, 0).
\end{equation}

\subsection{ART construction}

Let $h$ be a boundary condition satisfying (\ref{***}) on
$ \Gamma^{k_0}$.
 For $k\geq k_0+1$ define the following boundary condition on
 $ \Gamma^{k}$:
\begin{equation}\label{1}
\tilde{h}_x=\left\{\begin{array}{ll}
h_x, \ \ \mbox{if} \ \ x\in V^{k_0}\\[2mm]
0, \ \ \mbox{if} \ \ x\in V^k\setminus V^{k_0},\\
\end{array}\right.
\end{equation}
where $V^k$ denote the set of  vertices of $\Gamma^k$.
Namely,
 to each vertices of $V^{k_0}$ one adds $k-k_0$ successors with
vanishing value of the boundary condition.
It is obvious the b.c. $\tilde{h}$ satisfy the compatibility condition (\ref{***}).
In this way one constructs a new set of Gibbs measures that are extreme in the range  $1/k_0 < \theta < 1/ \sqrt{k}$ \cite{Ak}.

For the corresponding free energy, we have
\begin{equation}\label{art}
F_{\rm ART}(\beta, \tilde{h})=-{1\over \beta}\lim_{n\to \infty}{1\over |V^k_n|}
\left\{\left(|V^k_{n}|-|V^{k_0}_{n}|\right)\ln[2\cosh(\b J)]+
\beta\sum_{x\in V^{k_0}_{n}}a(x)\right\}.
\end{equation}
Since
$$\lim_{n\to \infty}{|V^{k_0}_{n}|\over |V^k_n|}={k-1\over k_0-1}\lim_{n\to\infty} {(k_0+1)k^n_0-2\over (k+1)k^n-2}=0,$$ by taking into account  $0 \leq a(x)\leq C_b$,
we get
$$0 \leq \sum_{x\in V_n^{k_0}}a(x)\leq |V_n^{k_0}|C_\b.$$
As a consequence,
$$\lim_{n\to \infty}{1\over |V^k_n|}
\sum_{x\in V^{k_0}_{n}}a(x)=0,$$
so that
\begin{equation}
F_{\rm ART}(\beta, \tilde{h})=-{1\over \beta}\ln[2\cosh(\b J)]=F_{\rm TI}(\beta,0).
\end{equation}

\subsection{Entropy}

To compute the entropy ${S}(\beta,h)=-{dF(\beta,h)\over dT}$, we first notice that equation (\ref{betah})
gives

\begin{equation}
h'(\beta )=\frac{1}{\beta '(h)}=J k \, \frac{\cosh (2 h) - \cosh \left(\frac{2 h}{k}\right)}{\sinh (2 h) - k \sinh \left(\frac{2 h}{k}\right)}.
\end{equation}
As a result of easy computations, we get the formula
\begin{equation}
\begin{split}
S(\beta ,h) & = \frac{1}{2} \ln \big[2 \cosh (2 h)+2 \cosh (2 \beta J )\big]  \\
& +
\beta J \, \frac{ k^2 \sinh (2 h) \frac{\cosh (2 h) - \cosh \left(\frac{2h}{k}\right)}
{\sinh (2 h) - k \sinh \left(\frac{2 h}{k}\right)} + \sinh (2 \beta J )}{\cosh (2 h)+\cosh (2 \beta J )}
\end{split}
\end{equation}
and
\begin{equation}
{S}(\beta,0)=\ln(2\cosh(\b J))-\b J\tanh(\b J).
\end{equation}

Let us mention that in case $k=2$, we get  by solving equation (\ref{f}) the following  value of $h_*$ as a function of the inverse temperature:
\begin{equation}\label{h}
\pm h_*={1\over 2}\ln\left[2^{-1}\left(e^{4\b J}-2e^{2\b J}-1\pm (e^{2\b J}-1)\sqrt{(e^{2\b J}+1)(e^{2\b J}-3)}\right)\right].
\end{equation}

The free energies  and entropies then read
\begin{eqnarray}
\label{tih}
F(\beta, h_*)&=&-{1\over \beta}\left(\ln(e^{2\b J}-1)+{1\over 2}\ln(e^{-2\b J}+1)\right),\\
S(\beta, h_*)&=&\ln(2\sinh(\b J))+{1\over 2}\ln(2\cosh(\b J))-\b J {3\cosh(2\b J)+1\over 2\sinh(2\b J)}.
\end{eqnarray}

\section{Periodic and weakly periodic Gibbs measures}

\subsection{A group representation of the Cayley tree}
Let $G_k$ be a free product of $k + 1$ cyclic groups of the second order with generators $a_1, a_2,\dots, a_{k+1}$,
respectively.

It is known that there exists a one-to-one correspondence between the set of vertices $V$ of the
Cayley tree $\Gamma^k$ and the group $G_k$.

To give this correspondence we fix an arbitrary element $x_0\in V$ and let it correspond to the unit element $e$ of the group $G_k$. Using $a_1,\dots,a_{k+1}$ we numerate nearest-neighbors of element $e$, moving by positive direction (see Fig. \ref{figure2}). Now we shall give numeration of the nearest-neighbors of each $a_i$, $i=1,\dots, k+1$ by $a_ia_j$, $j=1,\dots,k+1$. Since all $a_i$ have the common neighbor $e$ we give to it $a_ia_i=a_i^2=e$. Other neighbors are numerated starting from $a_ia_i$ by the positive direction. We numerate the set of all nearest-neighbors of each $a_ia_j$ by words $a_ia_ja_q$, $q=1,\dots,k+1$, starting from $a_ia_ja_j=a_i$ by the positive direction. Iterating this argument one gets
a one-to-one correspondence between the set of vertices $V$ of the
Cayley tree $\Gamma^k$ and the group $G_k$.

\begin{figure}
  \includegraphics[width=11cm]{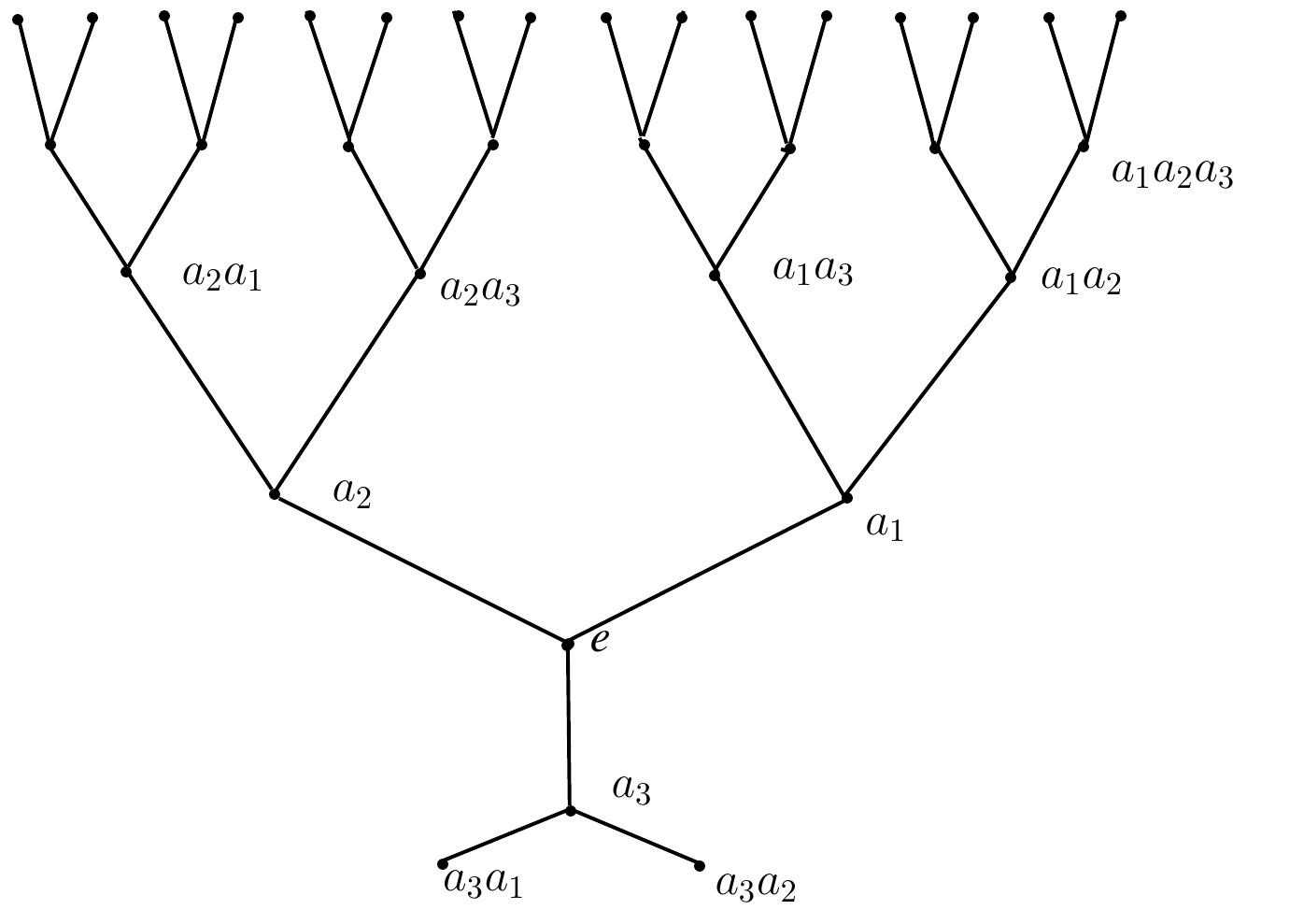}\\
  \caption{Some elements of group $G_2$ on Cayley tree of order two.}
\label{figure2}
\end{figure}

In the group $G_k$, let us consider the left (right) shift transformations defined as follows.
For $g_0\in G_k$, let us set
\begin{equation}
T_g(h) =gh,\ \ (T_g(h) = hg), \ \ \mbox{for all}\ \ h\in G_k.
\end{equation}
The set of all left  (right) shifts in $G_k$ is isomorphic to the group $G_k$.

\subsection{Periodic boundary conditions}

In this subsection, we consider periodic solutions
of (\ref{***}) and use
the above
group structure of the Cayley tree.

\begin{defn} Let ${\tilde G}$ be a normal subgroup of the group $G_k$. The set $h = \{h_x: x\in G_k\}$
 is said to be ${\tilde G}$-periodic if $h_{yx} =h_x$ for any $x\in G_k$ and $y\in {\tilde G}$.
 \end{defn}

Let

$$G^{(2)}_k = \{x\in G_k: \, \mbox{the length of word} \, x \, \mbox{is even}\}.$$
 Note that $G^{(2)}_k$ is  the  set of even vertices (i.e. with even distance to the root). Consider the boundary conditions
 ${h}^{\pm}$ and ${h}^{\mp}$:

\begin{equation}\label{PBC}
{h}_x^{\pm}=- {h}_x^{\mp} =
\left\{\begin{array}{ll}
h_*, \ \ \mbox{if} \ \ x\in G^{(2)}_k\\[2mm]
-h_*, \ \ \mbox{if} \ \ x\in G_k\setminus G^{(2)}_k,\\
\end{array}\right.
\end{equation}
and denote by $\mu^{(\mp)}, \mu^{(\pm)}$) the corresponding Gibbs measures.

The $\tilde{G}$- periodic solutions of equation (\ref{***}) are either translation-invariant ($G_k$-periodic) or $G^{(2)}_k$-periodic (see \cite{GR}), they are solutions to

\begin{equation}\label{ff}
u=kf(v,\theta), \ \ v=kf(u, \theta).
\end{equation}

In the ferromagnetic case only translation invariant b.c. can be found.
In the antiferromagnetic  case ($\theta \leq 0$) the system (\ref{ff}) has
 a unique solution $h=0$ if $ \theta\geq -1/k$, and three distinct solutions $h= 0$, ${h}^{\pm}$ and ${h}^{\mp}$ if
 $ \theta < -1/k$.

    Let us also recall that for the antiferromagnetic Ising model:
    \begin{itemize}
    \item[(1)] If $ \theta  \geq -1/k$,  $\mu_0$ is unique and extreme.

    \item[(2)] If $\theta <- 1/k$,   $\mu^{(\pm)}$ and $  \mu^{(\mp)}$, are extreme.

    \end{itemize}
 see \cite{Ge}.

 For periodic measures, we have
$$F_{\rm Per}(\b, h^{(\pm)}_*)= -{1\over 2\beta}\ln[4\cosh(\b J-h_*)\cosh(\b J+h_*)]=F_{\rm TI}(\b, h_*).$$

\subsection{Weakly periodic Gibbs measures.\label{SWP}} Let $G_k/\widehat{G}_k=\{H_1,...,H_r\}$  be a factor group, where
$\widehat{G}_k$ is a normal subgroup of index $r\geq 1$.

\begin{defn}\label{wp} A set  $h=\{h_x,x\in G_k\}$ is called
 $\widehat{G}_k$ - \textit{weakly periodic}, if
$h_x=h_{ij}$, for any $x\in H_i, x_{\downarrow}\in H_j$, where $x_{\downarrow}$ denotes the ancestor of $x$.
\end{defn}

Weakly periodic b.c.  $h$ coincide with periodic ones  if
 $h_x$ is independent of $x_{\downarrow}$.


Here, we will restrict ourself to the cases of index two and
recall that
any such subgroup has the form
$$H_A=\left\{x\in G_k:\sum\limits_{i\in A}\omega_x(a_i)-{\rm even} \right\},$$
where $\emptyset \neq A\subseteq N_k=\{1,2,\dots,k+1\}$, and $\omega_x(a_i)$ is the
number of $a_i$ in a word $x\in G_k$.
We consider  $A\ne N_k$: when  $A = N_k$ weak periodicity coincides
with standard periodicity.

 Let
$G_k/H_A=\{H_0, H_1\}$ be the  factor group, where  $H_0=H_A,
H_1=G_k\setminus H_A$.
 Then, in view of (\ref{***}), the
$H_A$-weakly periodic b.c.   has the form
\begin{equation}\label{wp5}
h_x=\left\{%
\begin{array}{ll}
    h_{1}, & {x \in H_0, \ x_{\downarrow} \in H_0}, \\[2mm]
    h_{2}, & {x \in H_0, \ x_{\downarrow} \in H_1}, \\[2mm]
    h_{3}, & {x \in H_1, \ x_{\downarrow} \in H_0}, \\[2mm]
    h_{4}, & { x \in H_1, \ x_{\downarrow}  \in H_1,}
\end{array}%
\right.\end{equation}
where  the $h_{i}$ satisfy the following equations:
\begin{equation}\label{wp6}
\left\{%
\begin{array}{ll}
    h_{1}=|A|f(h_{3},\theta)+(k-|A|)f(h_{1},\theta),\\[2mm]
    h_{2}=(|A|-1)f(h_{3},\theta)+(k+1-|A|)f(h_{1},\theta),\\[2mm]
    h_{3}=(|A|-1)f(h_{2},\theta)+(k+1-|A|)f(h_{4},\theta),\\[2mm]
    h_{4}=|A|f(h_{2},\theta)+(k-|A|)f(h_{4},\theta).
\end{array}%
\right.\end{equation}

For sake of simplicity, consider $k=4$ and $|A|=1$.
In this case
$$H_0=\{x\in G_k:\ \ \omega_x(a_1) \ \ \mbox{is even}\},$$
$$H_1=\{x\in G_k:\ \ \omega_x(a_1) \ \ \mbox{is odd}\}.$$
These sets are shown in Fig. \ref {WP}.

\begin{figure}
  \includegraphics[width=11cm]{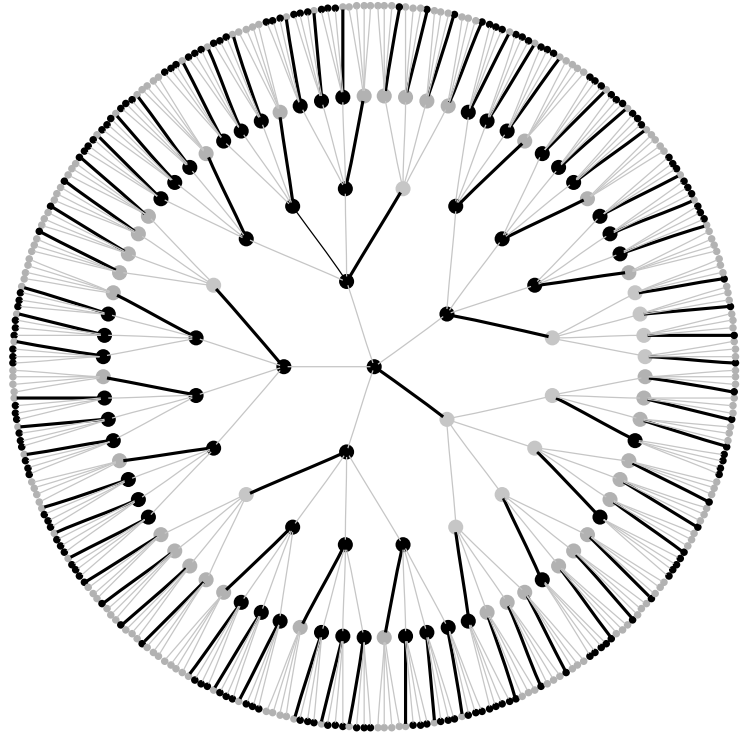}\\
  \caption{The sets $H_0$ (black vertices) and $H_1 $ (gray vertices).}\label{WP}
\end{figure}

Let us recall the following results of \cite{RR}:

There exists a critical value $\alpha_{\rm cr}$ $(\approx 0,1569)$ of  $\alpha=e^{-2\b J}$
such that:
  \begin{itemize}
    \item[(1)] If $\alpha > \alpha_{\rm cr}$, there exists a unique  weakly periodic state
 $\mu_0$,

    \item[(2)] If $\alpha = \alpha_{\rm cr}$,  there are three distinct weakly periodic states
$\mu_0, \mu^-_1, \mu^+_1$.

 \item[(3)] If $0\leq \alpha < \alpha_{\rm cr}$,  there are five distinct weakly periodic states
 $\mu_0, \mu^-_1, \mu^+_1, \mu^-_2, \mu^+_2$.
\end{itemize}
These measures correspond to solutions  of system  (\ref{wp6}) on the invariant set
$h_1=-h_4$, $h_2=-h_3$, that is to solutions of:
\begin{equation}\label{wpj7}
\left\{%
\begin{array}{ll}
    h_{1}=3 f(h_{1},\theta)-f(h_{2},\theta),\\[2mm]
    h_{2}=4 f(h_{1},\theta).\\[2mm]
    \end{array}%
\right.\end{equation}
More results about weakly periodic Gibbs measures can be found in  \cite{RR}, \cite{RR1}.

Denote
$$\mathcal A_n=\left|\left\{\langle x, y\rangle\in L_n:\, x\in H_0,\,
y=x_{\downarrow}\in H_0\right\}\right|,$$
\begin{equation}\label{AB}
\mathcal B_n=\left|\left\{\langle x, y\rangle\in L_n:\, x\in H_0,\,
y=x_{\downarrow}\in H_1\right\}\right|,
\end{equation}
$$\mathcal C_n=\left|\left\{\langle x, y\rangle\in L_n:\, x\in H_1,\,
y=x_{\downarrow}\in H_0\right\}\right|,$$
$$\mathcal D_n=\left|\left\{\langle x, y\rangle\in L_n:\, x\in H_1,\,
y=x_{\downarrow}\in H_1\right\}\right|,
$$
where $L_n$ is the set of edges in $V_n$.

For weakly periodic b.c. (\ref{wpj7}) we have
$$F_{\rm WP}(\beta, h)=-{1\over 2\beta}\lim_{n\to \infty}{1\over |V_n|}
\left\{\left(\mathcal A_{n}+\mathcal D_{n}\right)\ln[4\cosh(\b J-h_1)\cosh(\b J+h_1)]+\right.$$ $$\left.
(\mathcal B_{n}+\mathcal C_{n})\ln[4\cosh(\b J-h_2)\cosh(\b J+h_2)]\right\}.$$

By Proposition \ref{l4} of the appendix  we obtain
\begin{multline}
F_{\rm WP}(\beta, h)
=-{1\over 2\beta}
\left\{
{4\over 5}\ln[4\cosh(J\b-h_1)\cosh(J\b+h_1)]\right.
\\
+
\left.
{1\over 5}\ln[4\cosh(J\b-h_2)\cosh(J\b+h_2)]
\right\}
 \end{multline}
In the case under consideration the proposition is straightforward. Indeed one immediately see, in view of definitions and Fig. \ref{WP}, that
$ (A_{n}+\mathcal D_{n})/(|V_n|-1)=4/5$ and
$ (B_{n}+\mathcal C_{n})/(|V_n|-1)=1/5$ when $ n=1$; the induction is  trivial.

The equation
\begin{equation}\label{rr}
h_1=3f(h_1,\theta)-f(4f(h_1,\theta),\theta)
\end{equation}
that solves the system (\ref{wpj7})  (with
 $h_2=4f(h_1,\theta)$)
can then  be  reduced to:
\begin{equation}\label{rrx}
\alpha^2\xi^3-\alpha \xi^2-2\alpha^2\xi+\alpha+1=0,
\end{equation}
which
has two solutions $\xi_1$ and $\xi_2$
when  $0<\alpha<\alpha_{\rm cr}$ (see \cite{RR}).

The free energy then reads:
\begin{multline}
 F_{\rm WP}(\beta, h)=\nonumber
 \\
 -{1\over 10\beta}\left\{4\ln\left[2\cosh(2\b J)+{2\xi\cosh(2\b J)-4\over 2\cosh(2\b J)-\xi} \right]
 +
\ln[2\cosh(2\b J)+\xi^4-4\xi^2+2]\right\}.
\end{multline}

The  solution $\xi_1$ is given by
$$
\xi _1=\frac{1}{3} \left(e^{2 \beta }+\frac{2^{1/3} \left(6+e^{4 \beta }\right)}{U}+2^{-1/3} U \right),
$$
where
$$
U=\left(-9 e^{2 \beta }-27 e^{4 \beta }+2 e^{6 \beta }+3 \sqrt{-96-39 e^{4 \beta }+54 e^{6 \beta }+69 e^{8 \beta }-12 e^{10 \beta}}\right)^{1/3}.
$$
The corresponding free energy reads
$$
F_{\rm WP}(\beta) = -\frac{1}{10 \beta } \ln \frac{\left(-e^{-2 \beta }+e^{2 \beta }\right)^8 \left(2+e^{-2 \beta }+e^{2 \beta }
-4\left(e^{4 \beta }\frac{W}{6}\right)^2+\left(e^{4 \beta }\frac{W}{6}\right)^4\right)}{\left(e^{-2 \beta }+e^{2 \beta }- e^{4 \beta } \frac{W}{6}\right)^4},
$$

where
$$
W=\left(2 e^{-2 \beta }+\frac{V}{U}+2^{2/3} e^{-4 \beta } U\right)
\\
\quad
\text{with}
\quad
V=2^{4/3} e^{-4 \beta } \left(6+e^{4 \beta }\right).
$$
The plot is given in Fig. \ref{WPPlot} from which we observe the strict inequalities
\begin{equation}
F_{\rm TI}(\b, h_*)
<F_{\rm WP}(\b, h)
<F_{\rm TI}(\b, 0).
\end{equation}
  in the range  $\alpha \leq \alpha_{\rm cr}$.

\begin{figure}[h]
  \includegraphics[width=11cm]{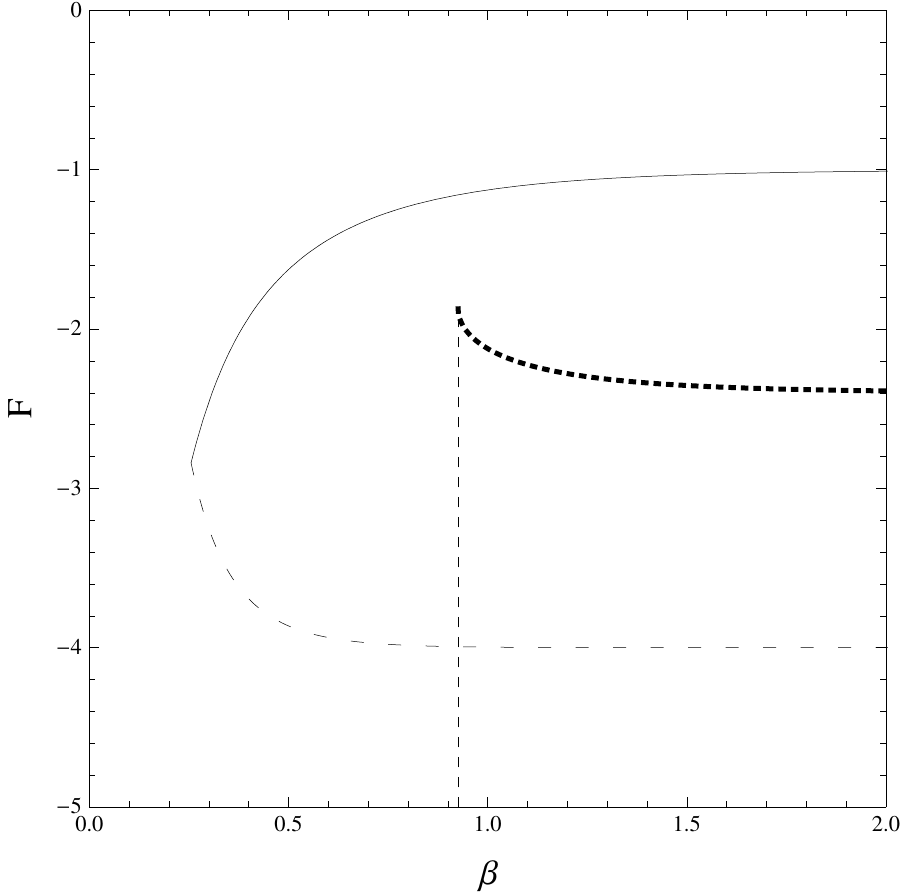}\\
  \caption{The free energies $F_{\rm WP}(\b)$ (dotted line) for $\alpha \leq \alpha_{\rm cr}$ together with the previous free energies   $F_{\rm TI}(\b, 0)$ (solid line),
and $F_{\rm TI}(\b, h_*)$ (dashed line). Here $J=1$ and $k=4$.}\label{WPPlot}
\end{figure}

The results for weakly periodic boundary conditions have to be compared with the inequalities given recently in  \cite{GRS}; the weakly periodic b.c.\  (\ref{wpj7})  corresponding to the so-called dimer covering in  \cite{GRS}. There,  the inequalities are easier to catch with cluster expansion method in mind, the condition on the   temperature is more restrictive,
and the  free energies cannot be expressed explicitly.

\section*{Appendix: Density of edges in a ball}

\renewcommand{\theequation}{A.\arabic{equation}}
\renewcommand{\thesection}{A}
\setcounter{equation}{0}


 In this appendix we consider a group representation of a Cayley tree and
its partition with respect to an arbitrary subgroup of index two. This partition gives a 2-vertex-coloring
on Cayley tree, which then gives 4-edge-coloring, say, colors $i=1,2,3,4$. We fix a root of the Cayley tree and
give explicit formulas for number $\mathcal A_{n,i}$ of edges with color $i$ in a ball $V_n$ of radius $n$ with
the center at the root. Moreover, we compute the $\lim_{n\to\infty}(\mathcal A_{n,i}/|V_n|)$ for each $i=1,2,3,4$.

We will use the notation of Subsection  \ref{SWP}
and let
$$\a_n=\left|\left\{\langle x, y\rangle\in L_n\setminus L_{n-1} :\, x\in H_0,\,
y=x_{\downarrow}\in H_0\right\}\right|.$$
$$\b_n=\left|\left\{\langle x, y\rangle\in L_n\setminus L_{n-1} :\, x\in H_0,\,
y=x_{\downarrow}\in H_1\right\}\right|.$$
$$\c_n=\left|\left\{\langle x, y\rangle\in L_n\setminus L_{n-1} :\, x\in H_1,\,
y=x_{\downarrow}\in H_0\right\}\right|.$$
$$\d_n=\left|\left\{\langle x, y\rangle\in L_n\setminus L_{n-1} :\, x\in H_1,\,
y=x_{\downarrow}\in H_1\right\}\right|,$$
for  $A=\{1,2,3,...,j\}$ and  $1\leq j\leq k+1$.


Let  $M$ be the set of all unit balls with vertices in $V$ and let $S_1(x)$ denotes the set of all nearest neighbors of $x$. For $b\in M$ the center of $b$ is denoted by $c_b$.

\begin{lemma}\label{l1} If $c_b\in H_0$, then
$$
\left|\{x\in S_1(c_b): x\in H_1\}\right|=j,\ \  \left|\{x\in S_1(c_b): x\in
H_0\}\right|=k-j+1.
$$
If $c_b\in H_1$, then
$$
\left|\{x\in S_1(c_b): x\in H_1\}\right|=k-j+1, \ \ \left|\{x\in S_1(c_b): x\in
H_0\}\right|=j.
$$
\end{lemma}
\begin{proof}
We have $S_1(c_b)=\{c_ba_p : p=1, 2, \dots, k+1\}.$

If $c_b\in H_0$, (the case $c_b\in H_1$ is similar) then $c_ba_p\in H_1$ for
$p=1,2,...,j$ and $c_ba_p\in H_0$ for $p=j+1,j+2,...,k+1,$ i.e. we have
$$ \left|\{x\in S_1(c_b): x\in H_1\}\right|=j, \ \  \left|\{x\in S_1(c_b):
x\in H_0\}\right|=k-j+1.
$$
\end{proof}

Consider $b=V_1\in M$ with the center $x^0=e\in H_0$, then in
$W_1$ we have $j$ vertices which belong to $H_1$, and $k-j+1$
vertices which belong in $H_0$, consequently,
$$\a_1=k-j+1, \ \  \b_1=0,  \ \ \c_1=j, \ \ \d_1=0.$$

\begin{lemma}\label{l2} For any $n\in N$ the following recurrence system hold
\begin{equation}\label{m1}
\left\{\begin{array}{llll}
    \a_{n+1}=(k-j)\a_n+(k-j+1)\b_n \\[2mm]
    \b_{n+1}=(j-1)\c_n+j\d_n \\[2mm]
    \c_{n+1}=(j-1)\b_n+j\a_n \\[2mm]
    \d_{n+1}=(k-j)\d_n+(k-j+1)\c_n,
\end{array}
\right.
\end{equation}
with initial values $\a_1=k-j+1$, $\b_1=0$, $\c_1=j$, $\d_1=0$.
\end{lemma}
\begin{proof} By Lemma \ref{l1}, an edge $\langle x, y\rangle\in L_n\setminus L_{n-1}$ with $x\in H_0,\,
y=x_{\downarrow}\in H_0$ has $(k-j)$ neighbor edges $\langle z, x\rangle\in L_{n+1}\setminus L_{n}$ with $z\in H_0,\,
x=z_{\downarrow}\in H_0$. An edge $\langle z, t\rangle\in L_n\setminus L_{n-1}$ with $z\in H_0,\,
t=z_{\downarrow}\in H_1$ has $(k-j+1)$ neighbor edges $\langle u, z\rangle\in L_{n+1}\setminus L_{n}$ with $u\in H_0,\,
z=u_{\downarrow}\in H_0$. Moreover, it is easy to see that only $\a_n$ and $\b_n$ have contribution to $\alpha_{n+1}$. Hence we have $\a_{n+1}=(k-j)\a_n+(k-j+1)\b_n$. Other equations of the system (\ref{m1}) can be obtained by a similar way.
\end{proof}

\begin{rk}\label{r1} For $j=k+1$ by Lemmas \ref{l1} and \ref{l2} we get $\alpha_n=\delta_n=0$, for any $n\geq 1$ and
$$\beta_n=\left\{\begin{array}{ll}
0, \ \ \mbox{if} \ \ n=2m-1\\[2mm]
(k+1)k^{2m-1}, \ \ \mbox{if} \ \ n=2m
\end{array}
\right., \ \ m=1,2,\dots
$$
$$\gamma_n=\left\{\begin{array}{ll}
0, \ \ \mbox{if} \ \ n=2m\\[2mm]
(k+1)k^{2(m-1)}, \ \ \mbox{if} \ \ n=2m-1
\end{array}
\right., \ \ m=1,2,\dots
$$
So in the sequel of this section we consider $j$ as $1\leq j\leq k$.
\end{rk}

\begin{lemma}\label{l3} For $\a_n$ we have
\begin{equation}\label{m2}
\a_{n+2}=j(k-j+1)|W_n|+(k-2j)\a_{n+1}-(2j-1)\a_n-k\a_{n-1},  \ \ n\geq 2,
\end{equation}
with initial values
\begin{equation}\label{m3}
\a_1=k-j+1, \ \ \a_2=(k-j)(k-j+1), \ \
 \a_3=\left((k-1)^2+j(j-1)\right)(k-j+1)
 \end{equation}
\end{lemma}

\begin{proof} The initial values follow from Lemma \ref{l2}.

By definitions of $\a_n$, $\b_n$, $\c_n$, $\d_n$ we have
\begin{equation}\label{m4}
\a_n+\b_n+\c_n+\d_n=|W_n|=k^{n-1}(k+1), \ \ n\geq 1.
\end{equation}
From (\ref{m1}) we get
\begin{equation}\label{m0}\begin{array}{lll}
\b_n=\frac{1}{k-j+1}(\a_{n+1}-(k-j)\a_n),\\[2mm]
\c_n=\frac{j-1}{k-j+1}(\a_{n}-(k-j)\a_{n-1})+j\a_{n-1},\\[2mm]
\d_n=\frac{1}{j}(\b_{n+1}-(j-1)\c_n).
\end{array}
\end{equation} Substituting these values in (\ref{m4}) and then simplifying we get (\ref{m2}).
\end{proof}
To find solution of (\ref{m2}) we denote
\begin{equation}\label{m5}
\a_n={q_n}{k^{n-2}}.
\end{equation}
From (\ref{m2}) we get
$$
k^nq_{n+2}=k^{n-1}(k+1)(k-j+1)j+(k-2j)k^{n-1}q_{n+1}-(2j-1)k^{n-2}q_n-k^{n-2}q_{n-1},
$$
dividing by $k^n$ we obtain
\begin{equation}\label{m6}
q_{n+2}=\frac{(k+1)(k-j+1)j}{k}+\frac{k-2j}{k}q_{n+1}-
\frac{2j-1}{k^2}q_{n}-\frac{1}{k^2}q_{n-1},
\end{equation}
with initial values
\begin{equation}\label{m7}
q_1=k(k-j+1),\,
q_2=(k-j)(k-j+1),\,
q_3=\frac{\left((k-1)^2+j(j-1)\right)(k-j+1)}{k}.
\end{equation}

In order to find solution to (\ref{m6}) first we rid $\frac{(k+1)(k-j+1)j}{k}$ by denoting
\begin{equation}\label{m8}
q_n=p_n+\frac{k(k-j+1)}{2}.
\end{equation}
Substituting (\ref{m8}) to (\ref{m6}) we get
\begin{equation}\label{m9}
p_{n+2}=\frac{k-2j}{k}p_{n+1}-\frac{2j-1}{k^2}p_{n}-\frac{1}{k^2}p_{n-1},
\end{equation}
with
\begin{equation}\label{m10}
p_1=\frac{1}{2}k(k-j+1), \, p_2=\frac{(k-2j)(k-j+1)}{2},\,
p_3=\frac{(k^2-4k+2+2j^2-2j)(k-j+1)}{2k}.
\end{equation}

The characteristic equation for (\ref{m9}) has the following form (setting $p_n=\lambda^n $):
$$
\lambda^{3}-\frac{k-2j}{k}\lambda^{2}+\frac{2j-1}{k^2}\lambda-\frac{1}{k^2}=0,
$$
which has solutions
\begin{equation}\label{la}
\lambda_1=-\frac{1}{k},\ \ \lambda_{2,3}=\frac{k-2j+1\pm \sqrt{(k-2j)^2-2(k+2j)+1}}{2k}.
\end{equation}

Then the general solution to (\ref{m9}) is
\begin{equation}\label{m11}
p_n=A_1\lambda_1^n+A_2 \lambda_2^n
+A_3\lambda_3^n,
\end{equation}
where the coefficients $A_1$, $A_2$, $A_3$ are determined by the initial conditions (\ref{m10}).

Using (\ref{m8}) and (\ref{m5}) we get

\begin{equation}\label{m11a}
\alpha_n={k-j+1\over 2}k^{n-1}+A_1\cdot {(-1)^n\over k^2}+{A_2\over k^2}\cdot \left(k\lambda_2\right)^n+
{A_3\over k^2}\cdot \left(k\lambda_3\right)^n.
\end{equation}

Then using (\ref{m0}) and (\ref{m11a}) one can find $\b_n$, $\c_n$ and $\d_n$.

We have
$$
{\mathcal A_n}=\sum_{m=1}^n \a_m={k-j+1\over 2(k-1)}\left(k^n-1\right)+{A_1\over 2k^2}\left((-1)^n-1\right)+$$
\begin{equation}\label{um}
{A_2\lambda_2k\over k^2(\lambda_2k-1)}\left((\lambda_2k)^n-1\right)+{A_3\lambda_3k\over k^2(\lambda_3k-1)}\left((\lambda_3k)^n-1\right).
\end{equation}

\begin{pro}\label{l4} For any $j=1, \dots, k$ and any fixed $q=0,1,2,\dots$ we have
$$\lim_{n\to\infty} {\a_{n-q}\over |V_n|}=\lim_{n\to\infty} {\d_{n-q}\over |V_n|}={(k-1)(k-j+1)\over 2(k+1)k^{q+1}}.$$
$$\lim_{n\to\infty} {\b_{n-q}\over |V_n|}=\lim_{n\to\infty} {\c_{n-q}\over |V_n|}={(k-1)j\over 2(k+1)k^{q+1}}.$$
$$\lim_{n\to\infty} {\mathcal A_{n-q}\over |V_n|}=\lim_{n\to\infty} {\mathcal D_{n-q}\over |V_n|}={{k-j+1}\over 2(k+1)k^q}.$$
$$\lim_{n\to\infty} {\mathcal B_{n-q}\over |V_n|}=\lim_{n\to\infty} {\mathcal C_{n-q}\over |V_n|}={j\over 2(k+1)k^q}.$$
\end{pro}
\begin{proof} It is easy to check that $|\lambda_2|<1$ and $|\lambda_3|<1$, i.e.
$$ \left|k-2j+1\pm
\sqrt{(k-2j)^2-2(k+2j)+1}\right|< 2k, \ \ \mbox{for any}\ \ 1\leq j\leq k. $$
Using these inequalities, formula (\ref{m11a}) and (\ref{vk})
we get $$\lim_{n\to\infty} {\a_{n-q}\over |V_n|}={(k-1)(k-j+1)\over 2(k+1)k^{q+1}}.$$
Now using this formula together with
(\ref{m0}) we obtain
$$\lim_{n\to\infty}{\beta_{n-q}\over |V_n|}={1\over k-j+1}\left(\lim_{n\to\infty}{\a_{n+1-q}\over |V_n|}-
(k-j)\lim_{n\to\infty} {\a_{n-q}\over |V_n|}\right)={j(k-1)\over 2(k+1)k^{q+1}}.$$
The formulae involving  $\gamma_n$ and $\delta_n$ are obtained in a similar way.

By  (\ref{um}) and
\begin{equation}\label{vk}
|V_n|={(k+1)\cdot k^n-2\over k-1}
\end{equation}
we get
$$\lim_{n\to\infty} {\mathcal A_{n-q}\over |V_n|}={{k-j+1}\over 2(k+1)k^q}.$$
From (\ref{m0}) we get
\begin{equation}\label{m0a}\begin{array}{lll}
\mathcal B_n=\frac{1}{k-j+1}\left(\mathcal A_n-\alpha_1+\alpha_{n+1}-(k-j)\mathcal A_n\right),\\[2mm]
\mathcal C_n=\frac{j-1}{k-j+1}\left(\mathcal A_{n}-(k-j)\mathcal A_{n-1}\right)+j\mathcal A_{n-1},\\[2mm]
\mathcal D_n=\frac{1}{j}\left(\mathcal B_n-\b_1+\beta_{n+1}-(j-1)\mathcal C_n\right),
\end{array}
\end{equation}
that allows to prove the remaining formulae.
\end{proof}

\begin{rk} By Proposition \ref{l4} it is clear that the values of $A_1$, $A_2$, $A_3$ do not
give any contribution to the equalities of the proposition. This is why we did not compute $A_1$, $A_2$, $A_3$.
 But one can obtain the numbers by the initial conditions (\ref{m10}) for $p_n$.
    For example, in the case  $k=4$, $j=1$ from (\ref{m11}) and  (\ref{m10}) we have
$$
p_n=A_1(-\frac{1}{4})^n+A_2\left(\frac{3-\sqrt{7}i}{8}\right)^n+A_3\left(\frac{3+
\sqrt{7}i}{8}\right)^n,
$$
$$
p_1=8, \ \ p_2=4, \ \ p_3=1.
$$
The initial conditions give
$$
\left\{%
\begin{array}{llll}
    8=-\frac{A_1}{4}+A_2\frac{3-\sqrt{7}i}{8}+A_3\frac{3+\sqrt{7}i}{8}\\
    4=\frac{A_1}{16}+A_2\frac{1-3\sqrt{7}i}{32}+A_3\frac{1+3\sqrt{7}i}{32}\\
    1=-\frac{A_1}{64}+A_2\frac{-9-5\sqrt{7}i}{128}+A_3\frac{-9+5\sqrt{7}i}{128} \\
\end{array}%
\right.\Rightarrow
\left\{%
\begin{array}{llll}
   A_1=0\\
   A_2=4+\frac{20\sqrt{7}i}{7}\\
   A_3=4-\frac{20\sqrt{7}i}{7} \\
\end{array}%
\right..
$$

Consequently, (for $k=4, j=1$) we have
\begin{equation}\label{ma}
\a_n=2\cdot4^{n-1}+\frac{1}{4}\left(1-\frac{5\sqrt{7}i}{7}\right)\left(\frac{3+\sqrt{7}i}{2}\right)^n+
\frac{1}{4}\left(1+\frac{5\sqrt{7}i}{7}\right)\left(\frac{3-\sqrt{7}i}{2}\right)^n.
\end{equation}
Using (\ref{m0}) and (\ref{ma}) one can find $\b_n$, $\c_n$ and $\d_n$.
Moreover, we have
$$
{\mathcal A_n}=\sum_{m=1}^n \a_m
=\frac{2(4^n-1)}{3}+\frac{2\sqrt{7}i}{7}\cdot\frac{(3-\sqrt{7}i)^n-(3+\sqrt{7}i)^n}{2^n}.
$$
Note that $\alpha_n$ and $\mathcal A_n$ are natural numbers for any $n\geq 1$.
\end{rk}

\begin{rk} In the case $j=k+1$ by Remark \ref{r1} we get
 $$\lim_{n\to\infty} {\a_{n}\over |V_n|}=\lim_{n\to\infty} {\d_{n}\over |V_n|}=\lim_{n\to\infty} {\mathcal A_{n}\over |V_n|}=\lim_{n\to\infty} {\mathcal D_{n}\over |V_n|}=0.$$
$$\lim_{m\to\infty} {\b_{2m-1}\over |V_{2m-1}|}=\lim_{m\to\infty} {\gamma_{2m}\over |V_{2m}|}=0.$$
$$\lim_{m\to\infty} {\b_{2m}\over |V_{2m}|}=\lim_{m\to\infty} {\gamma_{2m-1}\over |V_{2m-1}|}={k-1\over k}.$$
$$\lim_{m\to\infty} {\mathcal B_{2m}\over |V_{2m}|}=\lim_{m\to\infty} {\mathcal C_{2m-1}\over |V_{2m-1}|}={k\over k+1}.$$
$$\lim_{m\to\infty} {\mathcal B_{2m-1}\over |V_{2m-1}|}=\lim_{m\to\infty} {\mathcal C_{2m}\over |V_{2m}|}={1\over k+1}.$$
\end{rk}

\section*{Appendix: Existence of the free energy}

As we have seen in the previous sections free energy exists for each known compatible boundary condition.
We note that $a(x)$ is bounded: $\beta^{-1}\ln 2\leq a(x)\leq C_\beta$. Hence limit (\ref{fe2}) is also bounded.
But the problem of convergence of (\ref{fe2}) is still open.
Here we shall give some conditions on $h$,
 under which the corresponding free energy exists.

Let $\pi=\{x^0=x_0<x_1<\dots\}$ be an infinite path. A function $h_x$ on the path $\pi$
is called monotone non increasing (non decreasing) if $h_{x_i}\geq h_{x_{i+1}}$, ($h_{x_i}\leq h_{x_{i+1}}$), $i=0,1,2,\dots$.

\begin{pro}\label{pc} Let $h=\{h_x,\, x\in V\}$ be a compatible b.c.
If on any infinite path starting at $x\in W_{n_0}$, $|h_x|$ is monotone non increasing (non decreasing), then the corresponding free energy $F(\beta, h)$ exists.
\end{pro}
\begin{proof} Using Stolz-Ces\'aro theorem we get
$$\lim_{n\to\infty}{1\over |V_n|}
\sum_{x\in V_{n}}a(x)=\lim_{n\to\infty}A_n, \ \ \mbox{with} \ \ A_n={1\over |W_n|}
\sum_{x\in W_{n}}a(x).$$
We shall show that $A_n$ is monotone for $n>n_0$.
We have
$$A_{n-1}-A_n={1\over |W_n|}\left(k\sum_{x\in W_{n-1}}a(x)-\sum_{x\in W_{n}}a(x)\right)=$$
$${1 \over |W_n|}\sum_{x\in W_{n-1}}\left(ka(x)-\sum_{y\in S(x)}a(y)\right)=
{1 \over |W_n|}\sum_{x\in W_{n-1}}\sum_{y\in S(x)}\left(a(x)-a(y)\right).$$
By monotonicity of $|h_x|$ and by evenness of $a(x)$ we notice that $a(x)-a(y)$ does not change sign for all $x,y$ with $x<y$. Thus $A_n$ is monotone and since it is a bounded sequence it has a limit.
\end{proof}

\section*{ Acknowledgements}

U. Rozikov thanks CNRS for support and  The Centre de Physique Th\'eorique De Marseille, France for kind hospitality during his visit (September-December 2012).

\end{document}